\documentclass[11pt]{article}
\usepackage{graphicx,verbatim}
\usepackage{amsmath,amsthm}
\usepackage{comment}
\usepackage{algpseudocode}
\usepackage{algorithm}
\usepackage{amsfonts}
\usepackage{subcaption}

\newtheorem{theorem}{Theorem}
\newtheorem{mydef}{Definition}

\title{Enabling Privacy-Preserving GWAS in Heterogenous Human Populations}
\author{Sean Simmons\,$^{1,2,3}$, Cenk Sahinalp\,$^{3,4}$, and\\  Bonnie Berger\,$^{1,2}$\footnote{to whom correspondence should be addressed: bab@mit.edu}}
\date{}

\begin{document}

\maketitle

{\noindent $^{1}$Department of Mathematics, $^{2}$Computer Science and Artificial Intelligence Laboratory, Massachusetts Institute of Technology, Cambridge, MA\\
$^{3}$ School of Computing Science, Simon Fraser University, Burnaby, BC, Canada and\\
$^{4}$ School of Informatics and Computing, Indiana University, Bloomington, IN }

\newpage

\section{Abstract}

The projected increase of genotyping in the clinic and the rise of large genomic databases has led to the possibility of using patient medical data to perform genome-wide association studies (GWAS) on a larger scale and at a lower cost than ever before. Due to privacy concerns, however, access to this data is limited to a few trusted individuals, greatly reducing its impact on biomedical research. Privacy-preserving methods have been suggested as a way of allowing more people access to this precious data while protecting patients. In particular, there has been growing interest in applying the concept of differential privacy to GWAS results. Unfortunately, previous approaches for performing differentially private GWAS are based on rather simple statistics that have some major limitations--in particular, they do not correct for population stratification, a major issue when dealing with the genetically diverse populations present in modern GWAS. To address this concern we introduce a novel computational framework for performing GWAS that tailors ideas from differential privacy to protect private phenotype information, while at the same time correcting for population stratification. This framework allows us to produce privacy-preserving GWAS results based on two of the most commonly used GWAS statistics: EIGENSTRAT and linear mixed model (LMM) based statistics. We test our differentially private statistics, PrivSTRAT and PrivLMM, on both simulated and real GWAS datasets and find that they are able to protect privacy while returning meaningful GWAS results.

\section{Introduction}

With the projected increase of genotyping in the clinic and the rise of large genomic databases, there has been increasing interest in using patient data to perform genome-wide association studies (GWAS) \cite{I2B2,STRIDE}. The idea is to allow doctors and researchers to query patient electronic health records (EHR) to see which diseases are associated with which genomic alterations, avoiding the high costs required to recruit and genotype patients for a standard GWAS. Using this valuable data, however, leads to major privacy concerns for patients \cite{I2B2priv}. These privacy concerns have led to tight regulations over who can use this patient data--often it is limited to individuals who have gone through a time consuming and burdensome application process. Various approaches have been suggested for overcoming this major bottleneck in biomedical research. In particular, there has been interest in using a technique known as differential privacy \cite{CRr2013} to allow researchers access to this genomic data \cite{iDash,UTDP,CSAc2013,EpDP,MOTIF,blind,privgene} while preserving privacy.

Privacy concerns are not the only difficulty facing modern GWAS. GWAS aim to find biologically meaningful associations between common alleles in the population and disease status. This task, however, is complicated by systematic differences between different human populations \cite{herit}. It is often the case that biologically meaningful mutations are inherited jointly with mutations that have no such meaning, leading to false GWAS hits. A classic example of this phenomenon is given by the lactase gene. This gene is responsible for the ability to digest lactose (such as in milk), and is more common in those of Northern European ancestry than those of East Asian ancestry. People from Northern Europe are also, on average, taller than those from East Asia. This would lead a naive statistical method to erroneously suggest that the lactase gene is related to height. Such confounding effects are a major problem that can render the results of a GWAS (particularly one with large sample size) nearly nonsensical \cite{popStrat}. In order to avoid this common problem, known as population stratification, various methods have been employed (EIGENSTRAT \cite{EigenStrat}, linear mixed models (LMMs) \cite{herit}, genomic control (GC) \cite{GC}, etc.). 

In this work, we introduce the first method that jointly addresses the population stratification and privacy issues that arise when using patient data to answer GWAS queries.

\subsection{Our Contribution}

Previous work on differentially private GWAS have completely ignored the problem of population stratification, greatly limiting its applicability in the real world \cite{PoRu}. To help remedy this deficiency, we focus on producing GWAS results that can handle population stratification while still preserving private phenotype information (e.g., disease status). In particular, we develop a framework that can turn commonly used GWAS statistics (such as LMM based statistics and EIGENSTRAT) into tools for performing privacy-preserving GWAS. We will demonstrate this approach on two such statistics, EIGENSTRAT \cite{EigenStrat} and LMM based statistics \cite{herit}. Our methods, denoted PrivSTRAT and PrivLMM respectivelly, use a modified form of differential privacy to protect private phenotype information (disease status) from being leaked while returning highly associated SNPs. 

In particular, our new privacy framework allows us to repurpose three previous differentially private methods for picking high scoring SNPs to the EIGENSTRAT and LMM settings. We develop new algorithms that make these methods tractable (a limitation of some of the most promising differentially private GWAS methods proposed previously \cite{CSAc2013,iDashneigh}). We compare these methods on real and synthetic data, showing that one method, referred to as the distance method, greatly outperforms the other two in terms of accuracy. Importantly, ours is the first method able to correct for population stratification while preserving privacy in GWAS results. This opens up the possibility of applying a differentially private framework to large, genetically diverse groups of patients (such as those present in EHR!). 

\subsection{Previous Work}

GWAS aim to determine which common single nucleotide polymorphisms (SNPs) in the population are associated with a given disease. Numerous techniques (including genomic control \cite{GC}, EIGENSTRAT \cite{EigenStrat}, and linear mixed models (LMM) \cite{herit}) have been suggested to deal with population stratification in GWAS. In recent years, there has been a growing interest in using LMMs for this task, thanks to improved algorithms \cite{EMMAX,FASTLMM,GEORGE,PoRu,multLMM}. Even still, EIGENSTRAT remains a common approach for dealing with population stratification in practice.

Interest in privacy-preserving genomic analysis is a bit more recent \cite{YAn2014,BKCe2013,mask,MADEYh2013,relat,ours,JPS}. In particular, numerous works \cite{HOMER,TKe2010,HEDNx2012,XBYYHXg2011,SGMEn2009} have shown that GWAS statistics can leak private information about participants. Differential privacy \cite{CRr2013} (see below) has been suggested as a possible solution to the privacy conundrum \cite{iDash,iDashneigh,UTDP,CSAc2013,EpDP,MOTIF,blind,privgene,JPS}. There has even been a competition, hosted by iDASH, to help come up with better methods for performing differentially private GWAS \cite{iDash}. Note that, although much of this research has been encouraging, there is still a long way to go \cite{pharmaDPWar}.

\section{Definitions and Notation}

\subsection{GWAS Revisited}

The aim of a genome-wide association study (GWAS) is to link SNPs in a study cohort to a disease of interest. This is done by taking a large cohort of individuals, genotyping them at common SNPs, and, for each SNP, performing a statistical test to see if that SNP is associated with the disease in question. Note that, as with most work on GWAS, we assume each SNP has exactly two alleles: a minor allele and a major allele.

Formally, we have a group of $n$ individuals genotyped at $m$ SNPs. Let $D$ be an $n$ by $m$ genotype matrix, where the $i$th entry in the $j$th row of $D$ is equal to the number of times the minor allele occurs in the $j$th individual at the $i$th SNP (for autosomal SNPs this number is in the set $\{0,1,2\}$). Details on dealing with missing genotypes are provided in the Appendix. Let $X$ be the $n$ by $m$ matrix obtained by mean centering and variance normalizing each column of the genotype matrix $D$. Let $x_i$ be the column of $X$ corresponding to SNP $i$. Similarly, let $y=(y_1,\ldots,y_n)\in\{0,1\}^n$ be a vector of phenotypes, where $y_j=1$ if the $j$th individual has the disease, $y_j=0$ otherwise. 

Given $X$ and $y$, we would like to figure out which SNPs are associated with the disease in question. Naively, one could use a simple statistical test to figure this out (allelic test statistic, pearson test, logistic regression, linear regression, etc.). These statistics, however, ignore the effects of population stratification and lead to many false positives. Luckily, there have been various methods created to overcome this issue. Here we will mainly focus on one, EIGENSTRAT \cite{EigenStrat}, briefly touching on LMM based association \cite{EMMAX}.

\subsection{EIGENSTRAT Revisited}

One of the most popular methods for overcoming population stratification is known as EIGENSTRAT. This method is based off the observation that the top few principle components (which is to say the top few eigenvectors of the genetic covariance matrix) of the genotype matrix encode information about population stratification.

Formally, the method applies an eigendecomposition to the $n$ by $n$ covariance matrix $XX^T$. EIGENSTRAT works by forming two new vectors, $y^*$ and $x_i^*$, where $y^*$ (respectively $x_i^*$) is given by mean centering $y$ (respectively $x_i$) and projecting the result onto the vector space orthogonal to the top $k$ eigenvectors of $XX^T$ ($k$ is a user defined parameter; we set $k=5$). Intuitively, this procedure for producing $y^*$ and $x_i^*$ can be thought of as removing the effects of population stratification. Having removed the population stratification, all that remains is to test if $y^*$ and $x_i^*$ are correlated. This is done using a $\chi^2$-distributed statistic:

\[\chi^2_{i}=\frac{(n-k-1)(x_i^*\cdot y^*)^2}{|x_i^*|^2|y^*|^2}\]

\section{Methods}

\subsection{Differential Privacy and Private Phenotypes}

Differential privacy \cite{CRr2013} is an approach to privacy introduced by the cryptographic community. In a nutshell, it promises that that a given statistic calculated on one dataset behaves like the same statistic calculated on any dataset that differs in exactly one individual. In our case, since we are focusing on protecting phenotype data, we use a slightly modified definition:

\begin{mydef}
Let $F$ be a random function that takes in a $n$ by $m$ genotype matrix, $D$, and an $n$ dimensional phenotype vector, $y$, and outputs $F(D,y)$, where the output is in some set $\Omega$. We say that $F$ is $\epsilon$-phenotypic differentially private for some privacy parameter $\epsilon>0$ if, for all genotype matrices $D$, all phenotype vectors $y,y'\in \{0,1\}^n$ such that $y$ and $y'$ differ in exactly one coordinate, and for all sets $S\subset \Omega$, we have that
\[P(F(D,y)\in S)\leq exp(\epsilon)P(F(D,y')\in S)\] 
\end{mydef}

This differs from the usual definition of differential privacy since we are assuming the genotype matrix $D$ is fixed. Intuitively, our definition says that the result returned by $F$ when a given individual has the disease is statistically indistinguishable from the result returned when they do not have the disease (which is to say, for any $S\subset \Omega$, testing if $F(D,y)\in S$ reveals negligible private phenotype information). This indistinguishability ensures that $F$ gives away negligible information about the private phenotype $y$. 

The parameter $\epsilon$ is a privacy parameter: the closer to $0$ it is the more privacy is ensured, while the larger it is the weaker the privacy guarantee. This means we would like to set $\epsilon$ as small as possible, but unfortunately this comes at the cost of having less accurate outputs \cite{CRr2013}. 

Note that this is a slightly weaker definition of privacy than previous works--it does not guarantee that information about whether or not someone participated in our study is hidden. That said, when dealing with EHR, knowing that someone participated is equivalent to knowing they have their genotype on record at the hospital, a fact that is unlikely to be private.

\subsection{PrivSTRAT: Privacy-Preserving EIGENSTRAT}

The differentially private GWAS literature has largely focused on three tasks: picking highly associated SNPs, estimating association statistics, and estimating the number of significantly associated SNPs in a study. Due to space constraints we only consider the first problem here, though our framework can easily accomplish the other tasks as well (see the Appendix for details). 

In order to pick high scoring SNPs in a privacy-preserving manner we modify three previous methods (a noise based one, a score based one, and a distance based one \cite{CSAc2013,UTDP}) to EIGENSTRAT. 

Our task is to return the top $m_{ret}$ scoring SNPs for some user defined parameter $m_{ret}$ while achieving $\epsilon$-phenotypic differential privacy--that is to say we want to return the locations of the $m_{ret}$ SNPs with largest $\chi_i^2$ values (Note that this is a slightly different setup than in standard GWAS, where $m_{ret}$ is not known ahead of time. A discussion of this point is provided in the Appendix.). In order to do this, note that, if we let $\mu_{i}=\frac{x_i^*}{|x_i|}$, then: 

\[\chi_{i}^2=\frac{(n-k-1)(\mu_{i}\cdot y)^2}{|y^*|^2}\]

Since $|y^*|$ does not change from SNP to SNP, we see that picking the top $m_{ret}$ SNPs using EIGENSTRAT is equivalent to picking the $m_{ret}$ SNPs with largest $|\mu_i\cdot y|$ values. In order to do this in a privacy-preserving way let

\[\Delta=\max_{j\in\{1,\ldots,n\}}\max_{S\subset \{1,\ldots,m\},|S|=m_{ret}}\sum_{i\in S}|\mu_{ij}|\]

Our modified version of the noise based method for picking high scoring SNPs \cite{CSAc2013} works by calculating, for each $i$, $s_i=|\mu_iy|+Lap(0,\frac{2\Delta}{\epsilon})$, where
\[Lap(0,\lambda)\propto exp(-\frac{|x|}{\lambda})\]
This method then returns the $m_{ret}$ SNPs with the largest value of $s_i$. 

Similarly, our modified score based method \cite{CSAc2013} works by picking $m_{ret}$ SNPs without repetition, where the probability of picking the $i$th SNP is proportional to $exp(\frac{\epsilon|\mu_iy|}{2\Delta})$. Both the noise and score method are $\epsilon$-phenotypic differentially private (proofs in the Appendix).

The final approach is known as the distance based method \cite{UTDP}. This works as follows: the user chooses a threshold $c$. The $i$th SNP is considered significant if $|\mu_iy|>c$, not significant otherwise (for example, $c$ might correspond to a p-value of $.05$ or $10^{-8}$). The neighbor distance for the $i$th SNP, denoted $b_i$, is the minimum number of individuals whose phenotypes need to be changed to change SNP $i$ from significant to not or vice versa. Formally:

\[b_i=b_i(c)=\min_{y'\in[0,1]^n, c=|\mu_i\cdot y'|}|y-y'|_0\]

where $|v|_0$ denotes the number of nonzero entries in the vector $v$. Note that $b_i=\min\{d_i(c),d_i(-c)\}$, where $d_i(c)=\displaystyle\min_{y'\in[0,1]^n, c=\mu_i\cdot y'}|y-y'|_0$.

In order to use this neighbor distance to pick high scoring SNPs, we first have to let $d^*_i=b_i$ for significant SNPs and $d^*_i=1-b_i$ for all other SNPs. The distance based method picks $m_{ret}$ SNPs without repetition, where the probability of picking the $i$th SNP is proportional to $exp(\epsilon\frac{d^*_i}{2})$. Previous work \cite{UTDP} implies that this mechanism is indeed $\epsilon$-phenotypic differentially private. The difficult part is calculating $d_i(c)$. Our main algorithmic achievement is to show that this can be done using Algorithm \ref{Samp2}.

Taken together, these methods for picking high scoring SNPs while preserving privacy are called PrivSTRAT.

\begin{algorithm}
\caption{Calculates the neighbor distance}\label{Samp2}
\label{alg1}
\begin{algorithmic}
\Require $y,\mu_i,c$
\Ensure Returns the neighbor distance, $d_i(c)$.

\State Let $\hat{u}_j=\max(\mu_{ij}(1-y_j),\mu_{ij}(0-y_{j}) )$
\State Let $\hat{l}_j=\min(\mu_{ij}(1-y_j),\mu_{ij}(0-y_{j}) )$

\State Let $i_1,\cdots,i_n$ be a permutation on $1,\ldots,n$ such that $\hat{u}_{i_1}\geq \cdots\geq \hat{u}_{i_n}$. Let $u_j=\hat{u}_{i_j}$ for all $j$.

\State Let $j_1,\cdots,j_n$ be a permutation on $1,\ldots,n$ such that $\hat{l}_{j_1}\leq \cdots\leq \hat{l}_{j_n}$. Let $l_k=\hat{l}_{j_k}$ for all $k$.

\State Let $U_k=\sum_{j=1}^{k}u_j$ and $L_k=\sum_{j=1}^{k}l_j$, $k=1,\cdots,n$.

\State Return $k$ such that $c\in [L_{k+1},L_{k})\cup (U_{k},U_{k+1}]$

\end{algorithmic}
\end{algorithm}

\subsection{PrivLMM: Privacy-Preserving LMM Association}

Note that the above framework can be applied to other GWAS statistics besides EIGENSTRAT. In particular, it can be applied to linear mixed models (LMM) \cite{herit}. LMMs rely on the null model given by $y=X\beta+\epsilon$, where $\epsilon\propto N(0,\sigma_e^2\textbf{I}_n)$ and $\beta\propto N(0,\frac{\sigma_g^2}{m}\textbf{I}_m)$ for some unknown parameters $\sigma_e$ and $\sigma_g$ (where $\textbf{I}_n$ is the $n$ by $n$ identity matrix). 

Here we consider a slight modification of the LMM based approach used in EMMAX \cite{EMMAX}. This approach uses maximum likelihood (ML) to estimate $\sigma_e$ and $\sigma_g$. We can then apply the Wald test to see if a given SNP is associated with our disease phenotype. More specially, if we let $K=\sigma_e^2\textbf{I}_n+\frac{\sigma_g^2}{m}XX^T$, then we get a $\chi^2$ distributed statistic 

\[\chi_{i,LMM}^2=\frac{(x_i^TK^{-1}(\textbf{I}_n-\frac{1}{n}\textbf{1}_n)y)^2}{x_i^TK^{-1}x_i}\]

where $\textbf{1}_n$ is the $n$ by $n$ matrix of all ones. As was the case with EIGENSTRAT, it is worth noting that, if $\mu_{i,LMM}=\mu_{i,LMM}(\sigma_e^2,\sigma_g^2)=\frac{x_i^TK^{-1}(\textbf{I}_n-\frac{1}{n}\textbf{1}_n)}{\sqrt{x_i^TK^{-1}x_i}}$, then $\chi_{i,LMM}^2=(\mu_{i,LMM}\cdot y)^2$. This implies that high scoring SNPs correspond to SNPs with large values of $|\mu_{i,LMM}\cdot y|$.

This allows us to apply the framework we used for PrivSTRAT to this LMM statistic, giving us a method, denoted PrivLMM, that is phenotypically differentially private. The one added complication is that we need to be able to calculate $\sigma_e$ and $\sigma_g$ in a privacy-preserving way, but this is easily done using the sample-and-aggregate framework \cite{LMMDP} (see the Appendix). 


\section{Results}

We show that, on a real GWAS dataset with reasonable choices of $\epsilon$ (around $1.0$ or $2.0$) and $m_{ret}$ ($m_{ret}\in\{3,5\}$), both PrivSTRAT and PrivLMM have near perfect accuracy when using the distance method developed above. In order to test our methods, we implemented both of them in python using the pysnptools library \cite{PYSNPTOOLS}. 

\subsection{Data}

We test PrivSTRAT and PrivLMM on a Rheumatoid Arthritis (RA) dataset, NARAC-1, from Plenge et al. \cite{Plenge}. After quality control it contained 893 cases and 1243 controls, and a total of 67623 SNPs to be considered. Since this dataset has fairly little population stratification we also tried PrivSTRAT on a simulated dataset with two subpopulations. This dataset and the code to produce it (based off Plink tools \cite{PLINK}) are available online.

\subsection{Accuracy of PrivSTRAT}

\begin{figure*}
\centering
 \begin{subfigure}[A]{0.48\textwidth}
              \includegraphics[width=\textwidth]{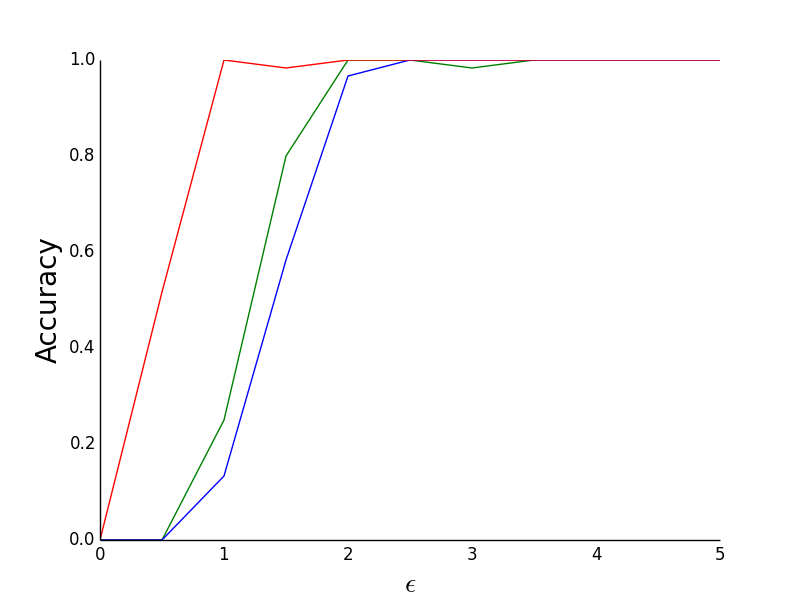}
		\caption{$m_{ret}=3$, Real GWAS}
        \end{subfigure}
        \begin{subfigure}[B]{0.48\textwidth}
               \includegraphics[width=\textwidth]{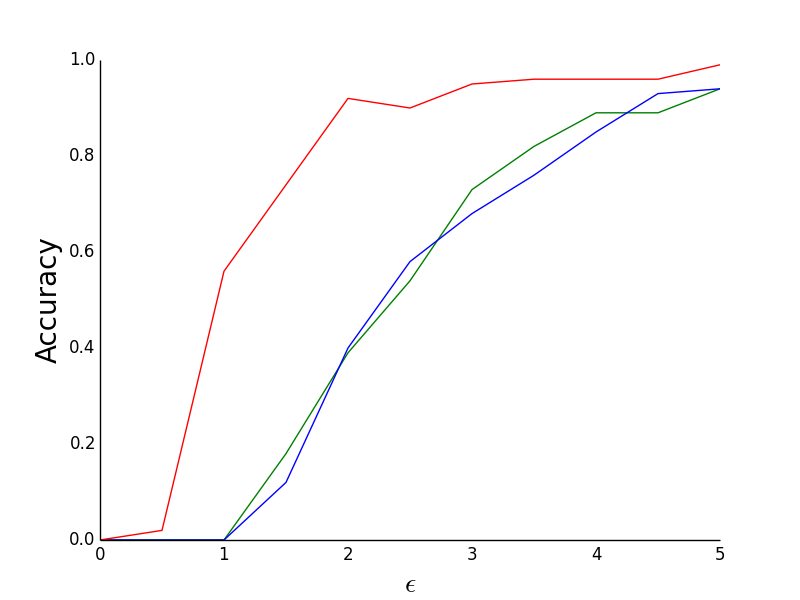}
		\caption{$m_{ret}=5$, Real GWAS}
        \end{subfigure}
       
        \begin{subfigure}[A]{0.48\textwidth}
              \includegraphics[width=\textwidth]{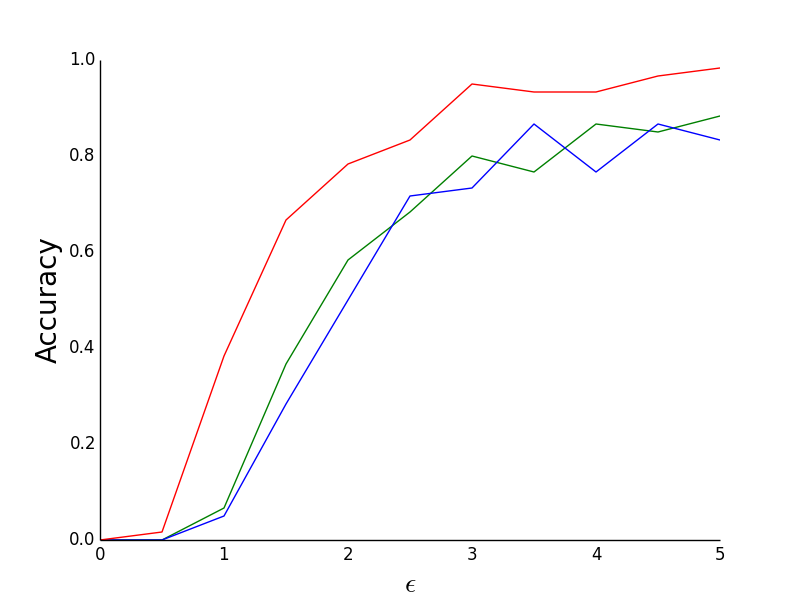}
		\caption{$m_{ret}=3$, Simulated GWAS}
        \end{subfigure}
        \begin{subfigure}[B]{0.48\textwidth}
               \includegraphics[width=\textwidth]{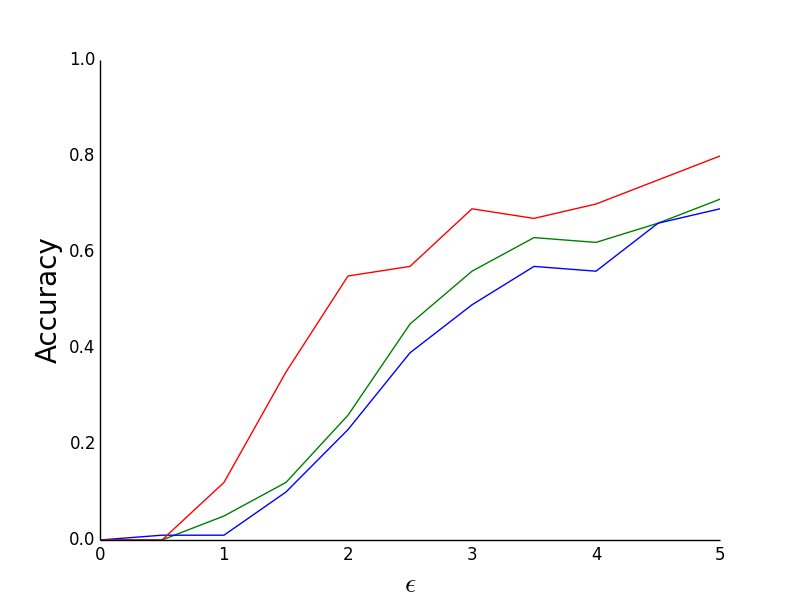}
		\caption{$m_{ret}=5$, Simulated GWAS}
        \end{subfigure}
\caption{We measure the accuracy (the percentage of the top SNPs correctly returned) of the three PrivSTRAT algorithms for picking top SNPs using score (blue), distance (red) and noise (green) based methods with $m_{ret}$ (the number of SNPs being returned) equal to (a) 3 and (b) 5  for our RA GWAS dataset, and with $m_{ret}$ equal to (c) 3 and (d) 5 for our simulated dataset, with varying values of the privacy parameter $\epsilon$. We see that, in all four cases, the distance based method outperforms the others.  These results are averaged over 20 iterations.}\label{Top}
\end{figure*}

We tested the accuracy of PrivSTRAT for picking high scoring SNPs. We tested each algorithm (noise, score, and distance based) for returning the top $m_{ret}$ SNPs, where $m_{ret}\in\{3,5\}$ (this choice is based off previous work \cite{CSAc2013}. Other values are explored in the Appendix.), and for various choices of the privacy parameter, $\epsilon$. The accuracy of the returned results (averaged over 20 trials) is measured by the percentage overlap between the returned results and the true results \cite{CSAc2013}. The results on the RA dataset are pictured in Fig \ref{Top}a and \ref{Top}b, while the results on the simulated data are shown in Fig \ref{Top}c and \ref{Top}d. We see that, as expected, as $\epsilon$ increases, accuracy increases. Moreover, we see that the noise and score based methods (in blue and green) do not perform as well as the distance based method (in red). This is not surprising, agreeing with previous work (our inclusion of the score and noise based methods is for the sake completeness). Moreover, on the real GWAS data we get near perfect accuracy for realistic values of $\epsilon$ (values around 1 or 2) \cite{Count}, accuracy that will increase as datasets grow.

\subsection{Runtime}

Though the privacy preserving methods add extra runtime to our method, the runtime is less than that required by EIGENSTRAT to find the top PCs. 

Note that, as in EIGENSTRAT, PrivSTRAT calculates the top PCs by performing singular value decomposition (SVD) on the normalized genotype matrix $X$. Note that our current implementation of PrivMAF uses a fast, approximate method for performing this SVD decomposition (details are in the appendix). This differs from the standard smartpca algorithm used by EIGENSTRAT (note that the newest version of EIGENSTRAT has also implemented a fast approximation similar to the one we use)

Therefore, in order to look at how the privacy preserving nature of PrivSTRAT affects runtime we looked at the runtime of PrivSTRAT using both the exact and approximate methods for calculating the SVD (Figure \ref{runtime}). More specifically, we ran PrivSTRAT on the RA dataset described above with $m_{ret}=3$, and looked at the amount of time taken by each step of the algorithm: calculating the SVD (using either exact (the smartpca algorithm included in EIGENSTRAT) or approximate methods), calculating the neighbor distance, and picking the SNPs. The results are an average over 10 trials. We see that the calculation of the exact SVD is (by far) the slowest of these steps, while even the approximate SVD calculation is only a factor of 2 faster than the slowest step in the privacy preserving algorithm.

Asymptotically, we see that the calculation of the exact neighbor distance is by far the most time consuming (running in time $O(n^2m)$), followed by the calculation of the neighbor distance ($O(nmlog(n))$) which is slightly slower than the approximate SVD calculation ($O(nmlog(n))$).

\begin{table}[!t]
\begin{center}
\caption{We compare the runtime of different steps of picking high scoring SNPs in a privacy preserving manner (measure in seconds). The runtime is calculated on the RA dataset, with $m_{ret}=3$, $k=5$, and averaged over 10 runs. We see that the smartpca method (the exact SVD method used in EIGENSTRAT, see text for details) takes the most time, followed by the calculation of neighbor distance. The approximate SVD method we use takes about half the time of calculating the neighbor distance.}\label{runtime}. 
\begin{tabular}{|l|l|l|l|l|}
\hline
Approx SVD & Exact SVD & Calculate $\mu$ &  Neighbor Distance & Pick SNPs \\\hline
14.37 seconds & 134.16 seconds & 8.60 seconds & 26.23 seconds & .25 seconds\\\hline
\end{tabular}
\end{center}
\end{table}

\subsection{Accuracy of PrivLMM}

\begin{figure*}
\centering
 \begin{subfigure}[A]{0.48\textwidth}
              \includegraphics[width=\textwidth]{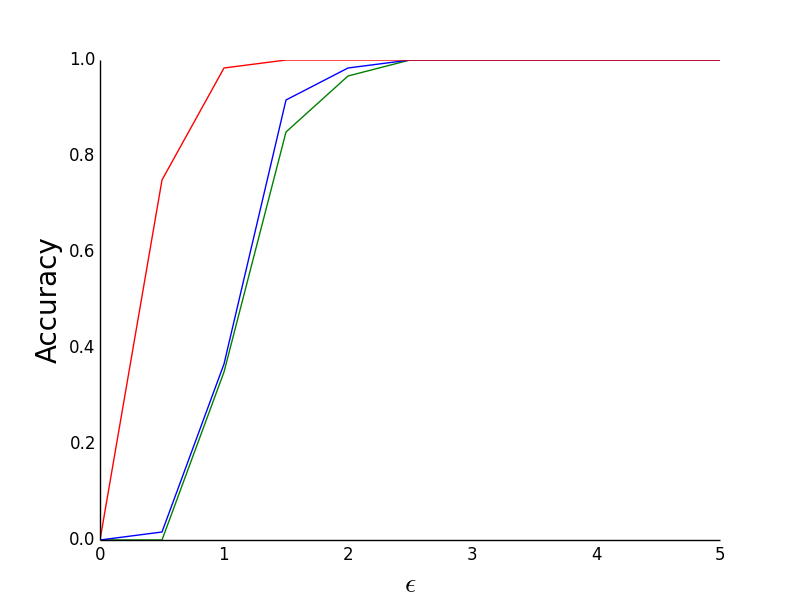}
		\caption{$m_{ret}=3$, Real GWAS}
        \end{subfigure}
        \begin{subfigure}[B]{0.48\textwidth}
               \includegraphics[width=\textwidth]{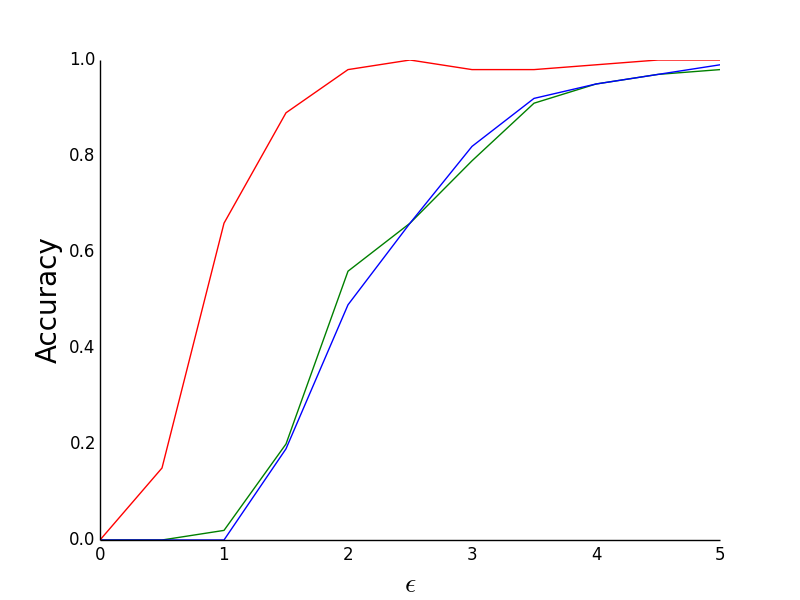}
		\caption{$m_{ret}=5$, Real GWAS}
        \end{subfigure}

\caption{We measure the accuracy (the percentage of the top SNPs correctly returned) of the three PrivLMM algorithms for picking top SNPs using score (blue), distance (red) and noise (green) based methods with $m_{ret}$ (the number of SNPs being returned) equal to (a) 3 and (b) 5  for our RA GWAS dataset, with varying values of the privacy parameter $\epsilon$. We see that, in both cases, the distance based method outperforms the others.  These results are averaged over 20 iterations.}\label{fake}
\end{figure*}

We also tested the accuracy of PrivLMM on our RA GWAS dataset (we do not include the simulated dataset due to space constraints). We used the same set up as for PrivSTRAT. The results are pictured in Fig \ref{fake}. We see that, as expected, as $\epsilon$ increases, accuracy increases, and that the noise and score based methods (in blue and green) do not perform as well as the distance based method (in red). Note that we used values of $\sigma_e$ and $\sigma_g$ calculated using FaST-LMM \cite{FASTLMM} software. In theory, it is preferable to use a differentially private approach to calculate $\sigma_e$ and $\sigma_g$. A method to do this, based on previous work \cite{LMMDP}, is given in the Appendix.

\section{Conclusion}

We have demonstrated we can perform privacy-preserving GWAS while correcting for the effects of population stratification without significant increase in running time. Note that the major computational bottleneck in our methods comes not from the privacy preserving component, but instead arises from the original statistics (from calculating the eigenvectors in EIGENSTRAT or inverting the matrix in the LMM based statistic). As such, our methods can exploit future computational advances in GWAS analysis. In particular, we are interested in modifying our method to take advantage of the computational advances introduced by Loh et al. \cite{PoRu} for LMM based association.

We would also like to extend our methods to settings where stronger privacy guarantees (beyond just protecting private phenotype data) are required. Another potential direction involves recent work showing that background knowledge about haplotypes \cite{blind} and population genetics \cite{JPS} can improve accuracy in privacy preserving genomic analysis. It would be of great interest to see if these approaches can be used to improve the accuracy of PrivSTRAT and PrivLMM.

In addition to improving privacy, recent theoretical work has shown that differential privacy can help prevent false positives due to overfitting in adaptive data analysis (looking at the data to decide which analysis techniques to use), overcoming a major problem in medical research \cite{overfit}. This line of inquiry opens up exciting possibilities for how our results might be used in the future.

Note that we are not advocating privacy preserving methods for all situations in which one might want to conduct a GWAS, but only when privacy concerns would make alternative approaches cumbersome or impossible.

It is our hope that our Priv suite of tools will be used to open up private genomics data to a much larger group of researchers. This access will give researchers new tools that can be used to produce novel hypotheses or validate old results in ways that are not currently possible due to privacy concerns. 

Availability: An implementation of our results and simulated data is available on our website, http://groups.csail.mit.edu/cb/PrivGWAS.

\newpage

\bibliographystyle{plain}
\bibliography{main}

\newpage

\section{Appendix}

\renewcommand{\thefigure}{S\arabic{figure}}

\setcounter{figure}{0}

\subsection{Proofs of correctness}

\begin{theorem}
The modified versions of the score and noise based methods for picking high scoring SNPs given in the manuscript are $\epsilon$-phenotypic differentially private.
\end{theorem}
\begin{proof}
The proofs are similar to those given in previous works \cite{CSAc2013}, except we use a score function where the score of returning SNPs $s_1,\cdots,s_{m_{ret}}$ equals $\displaystyle\sum_{i=1}^{m_{ret}}|\mu_{s_i}\cdot y|$. For completeness we give the details below.

To see that this is true for the score method, let $\mathbb{S}$ be the collection of all ordered sets of exactly $m_{ret}$ SNPs. Define the score function 
\[q:\mathbb{S}\times \{0,1\}^n\rightarrow \mathbb{R}\]
so that
\[q(s_1,\cdots,s_{m_{ret}},y)=\displaystyle\sum_{i=1}^{m_{ret}}|\mu_{s_i}\cdot y|\]
Note that, if $y,y'\in \{0,1\}^n$ differ in exactly one coordinate, then for any $s_1,\cdots,s_{m_{ret}}$ we have that:

\[|q(s_1,\cdots,s_{m_{ret}},y)-q(s_1,\cdots,s_{m_{ret}},y')|\leq \sum_{i=1}^{m_{ret}}|\mu_{s_i}\cdot (y-y')|\]
\[\leq \Delta\]

where $\Delta$ is defined as in the text. Therefore the result follows from the properties of the exponential mechanism \cite{ExpMech}.

Next consider the noise method. Again, using the fact that 
\[|q(s_1,\cdots,s_{m_{ret}},y)-q(s_1,\cdots,s_{m_{ret}},y')|\leq \sum_{i=1}^{m_{ret}}|\mu_{s_i}\cdot (y-y')|\]
\[\leq \Delta\]
the result follows from \cite{CSAc2013}.

\end{proof}

\begin{theorem}
Algorithm \ref{Samp2} returns the correct value of $d_i(c)$.
\end{theorem}
\begin{proof}

Let $U_k$, $L_k$, $l_k$ and $u_k$ be as in Algorithm \ref{Samp2}.

Assume that $y$ and $y'$ differ in at most $k$ coordinates, then

\[\mu_iy-\mu_iy'=\sum_{j,y_j\neq y_j'}\mu_{ij}(y_j-y_j')\leq -(l_1+\cdots+l_k)\]
so 
\[\mu_iy'\geq \mu_iy-\sum_{i=1}^kl_k=L_k\] 

Similarly

\[\mu_iy'\leq \mu_iy+\sum_{i=1}^ku_k=U_k\]

so if $d_i(c)\leq k$ than $L_k\leq c\leq U_k$. It is easy to see, however, that if $L_k\leq c\leq U_k$ than $d_i(c)\leq k$, so $d_i(c)=k$ if and only if $c\in [L_k,L_{k-1})\cup (U_{k-1},U_k]$. Therefore Algorithm \ref{Samp2} correctly calculates $d_i(c)$.

\end{proof}

\subsection{Details About the Distance Based Method}

Note that the distance based method for picking high scoring SNPs requires the choice of a boundary value, $c$. This value is a kind of baseline. Previous work, however, has shown that this arbitrary choice of $c$ can change the accuracy of the method \cite{CSAc2013}. 

In order to deal with this we use a slightly modified version of the distance based method. For a given $\epsilon$ and choice of $m_{ret}$, let $x_1,\cdots,x_m$ be a reordering of the list $|\mu_1 \cdot y|,\cdots,|\mu_m \cdot y|$ in decreasing order. Than we can choose $c$ so that

\[c=\frac{|x_{m_{ret}}|+|x_{m_{ret}+1}|}{2}+Lap(0.0,max_{i,j}\frac{|\mu_{i,j}|}{.1\epsilon})\]

We than run the distance based method with a privacy budget of $.9\epsilon$ and a boundary of $c$. This approach is still $\epsilon$-phenotypic differentially private, and removes some of the accuracy issues of previous approaches.

\subsection{Simulated dataset}

In order to produce simulated data, we used PLINK \cite{PLINK}. The code used to generate this data is available on our website. 

We generated two populations of individuals. For each set we first used plink to choose the MAF for 10000 SNPs, each uniformly at random from [.05,.5]. 9900 of the SNPs had no effect on phenotype, 100 had an odds ratio of 1.1. We then generated 5000 people from each of the populations, half of whom where cases, the other half controls. We then combined these two populations to produce our simulated dataset.

The code to do this is present online, as is the simulated data generated in this way.

\subsection{Estimating Heritability}

Another issue to consider is the estimation of $\sigma_e$ and $\sigma_g$ in PrivLMM. This, however, can be done using a sample-and-aggregate based framework \cite{LMMDP}. In particular, the works by choosing some integer $K>1$, and dividing the set of participants into $K$ disjoint sets of equal size. On each of these subsets we can estimate $h^2=\frac{\sigma_g^2}{\sigma_e^2+\sigma_g^2}$ using FaST-LMM \cite{FASTLMM}, GCTA \cite{GCTA} or a similar tool (our implementation uses FaST-LMM). This gives us $K$ estimates of $h^2$, namely $h_1^2,\ldots,h_K^2$. Let $\tilde{h}^2$ be the average of these $K$ values. Our $\epsilon$-differentially private estimate of $h^2$ is then given by calculating $\tilde{h}^2+Lap(0,\frac{1}{K\epsilon})$ and rounding the result to the interval $[0,1]$. 

Next we want to use the same framework to estimate $\sigma_e^2$. Note, however, that this would require a bound on $\sigma_e^2$. Note that $\sigma_e^2\leq Var(y)$, and that we can get a $\epsilon$-differentially private estimate $v_{dp}$ of $Var(y)$ easily using the laplacian mechanism. Then we can easily apply the sample-and-aggregate methodology to $\max\{v_{dp},\sigma_e^2\}$ to get an $\epsilon$-differentially private estimate.  Since $\sigma_g^2=\sigma_e^2(\frac{1}{1-h^2}-1)$ this allows us to get a $3\epsilon$-differentially private estimate of $(\sigma_e^2,\sigma_g^2)$. Note that this method relies on a very general methodology, and so it seems likely much more accurate results can be obtained with a little work. 

\subsection{Estimating $\chi^2$}

In addition to picking high scoring SNPs, we would like to estimate the associated $\chi^2$-statistic for EIGENSTRAT. In particular, assume we want to get an estimate of $\chi_i^2$ for a given SNP $i$. Note, however, that 
\[\chi_i^2=\frac{(n-k-1)(\mu_i \cdot y)^2}{|y^*|^2}\]
so it suffices to get estimates of both $\mu_i \cdot y$ and $|y^*|$ that are $\frac{\epsilon}{2}$-phenotypic differentially private and combine the results. This can be done easily, however, using the Laplacian mechanism \cite{CRr2013}, which gives us $\mu_i\cdot y+Lap(0,\frac{2max_{j}|\mu_{i,j}|}{\epsilon})$ and $|y^*|+Lap(0,\frac{2}{\epsilon})$ as estimates.

The approach taken by PrivLMM is almost identical, except there is no need to estimate $|y^*|^2$, only $\sigma_e$ and $\sigma_g$ (see above).

\subsection{Estimating The Number of Significant SNPs}\label{numsig}

We would also like to estimate the number of significant SNPs in a differentially private way for EIGENSTRAT--that is to say estimate the number of SNPs with $\chi_i^2>c$ for some user defined $c$ (often corresponding to a particular p-value cut off). This is equivalent to estimating the number of SNPs with $|\mu_iy|\geq |y^*|\sqrt{\frac{c}{n-k-1}}$. 

In order to do this we first calculate a $.1\epsilon$-phenotypic differentially private estimate of $|y^*|\sqrt{\frac{c}{n-k-1}}$, denoted $c_{dp}$, using the Laplacian mechanism:

\[c_{dp}=(|y^*|+Lap(0,\frac{10}{\epsilon}))\sqrt{\frac{c}{n-k-1}}\]

Since we know how to calculate $b_i(c_{dp})$ (see above) it is easy to apply the method of Johnson and Shmatikov \cite{UTDP} to get a $.9\epsilon$-phenotypic differentially private estimate of the number of SNPs with $|\mu_iy|\geq c_{dp}$, which is returned to the researcher. The result is an $\epsilon$-phenotypic differentially private estimate of the number of significant SNPs. Note that the choice of $.1\epsilon$ and $.9\epsilon$ are arbitrary, and can be played around with for better results.

A similar method works for PrivLMM.

\subsection{Accuracy of $\chi^2$}

\begin{figure}
\centering
 \begin{subfigure}[A]{0.48\textwidth}
              \includegraphics[width=\textwidth]{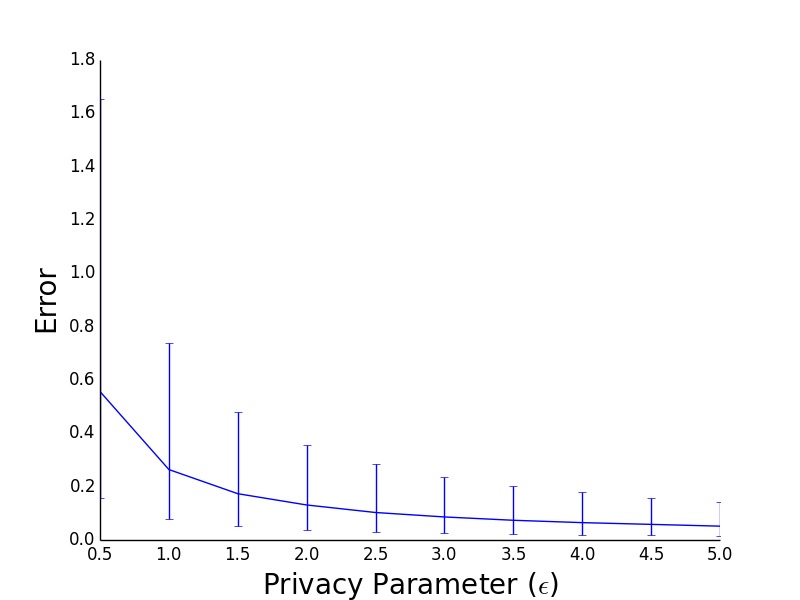}
		\caption{Real GWAS Data}
        \end{subfigure}
        \begin{subfigure}[B]{0.48\textwidth}
               \includegraphics[width=\textwidth]{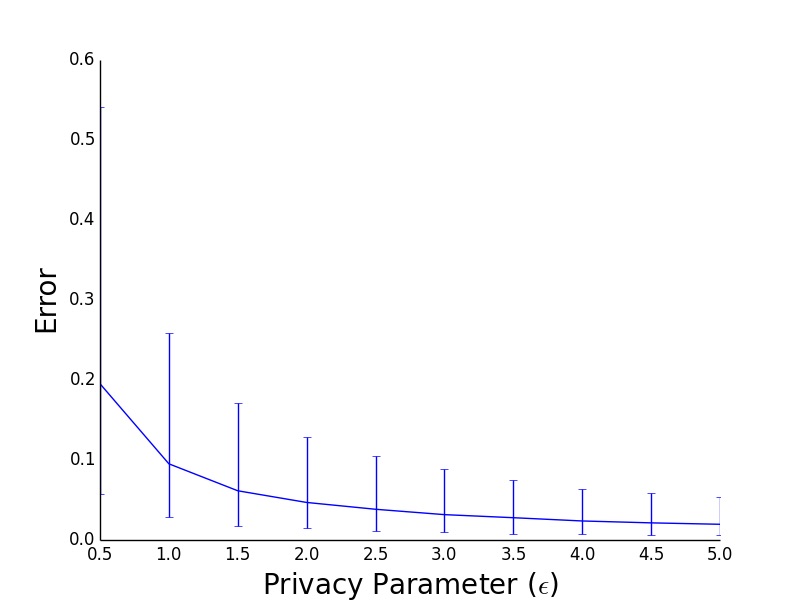}
		\caption{Simulated GWAS Data}
        \end{subfigure}
       
       \caption{Here we look at the accuracy of our mechanisms for approximating the EIGENSTRAT statistic on real (a) and simulated (b) GWAS data for various privacy parameters, $\epsilon$. We plot the median error over all SNPs, with error bars representing the 25\% and 75\% quantiles. As expected we see that accuracy increases as $\epsilon$ increases (aka as privacy decreases).}\label{chi}
\end{figure}

We also looked at the accuracy of PrivSTRAT's estimated $\chi^2$ value. Fig \ref{chi}(a) demonstrates our method on real GWAS data, Fig \ref{chi}(b) on simulated data. We plot the median error over all SNPs, with error bars representing the 25\% and 75\% quantiles. As expected our method increase in accuracy as $\epsilon$ increases.

\subsection{Difference From Standard GWAS}

The privacy preserving framework we introduce here is slightly different than that taken in standard GWAS. In particular, in standard GWAS the quantity $m_{ret}$ (the number of SNPs to be returned) is not known ahead of time. Instead, the user sets some p-value and gets back a list of all SNPs whose p-value is less than that boundary.

If one wants to perform such a study, they can use the method introduced above for calculating the number of significant SNPs (Section \ref{numsig}), and then use the returned value as $m_{ret}$. In order to ensure accuracy, however, it seems more reasonable to choose a small $m_{ret}$ ahead of time. This ensures the accuracy of the GWAS on the highest scoring SNPs, even if it comes at a cost to some SNPs near the p-value threshold of interest.

\subsection{Large Values of $m_{ret}$}

Our experiments show that, for small $m_{ret}$ ($m_{ret}\leq 5$ or so) that our methods are reasonably accurate. It turns out, however, that like previous approaches to differentially private GWAS, these methods do not always scale to large $m_{ret}$ when the number of individuals is small. This is shown for PrivSTRAT on the RA dataset in Fig \ref{largeM}. We see that, though accuracy is comparable to methods that do not correct for population stratification, it is still not as useful as we would like. Luckily, this accuracy should greatly increase as $n$ gets larger.

It is also worth asking if accuracy is the best measure of utility for our method. In particular, using accuracy to measure utility ignores the difference between returning SNPs that score almost as high as the top scoring SNPs versus returning low scoring SNPs. Moreover, GWAS assumes we are using SNPs to tag nearby regions of the genome. This implies that returning a SNP that is near a high scoring SNP can also be useful. Using accuracy as the measurement, however, ignores this as well. Therefore, in order to decide if our method is useful for larger $m_{ret}$ values, we should first decide exactly what makes a given result preferable to another.

\begin{figure}
\centering
 \begin{subfigure}[A]{0.48\textwidth}
              \includegraphics[width=\textwidth]{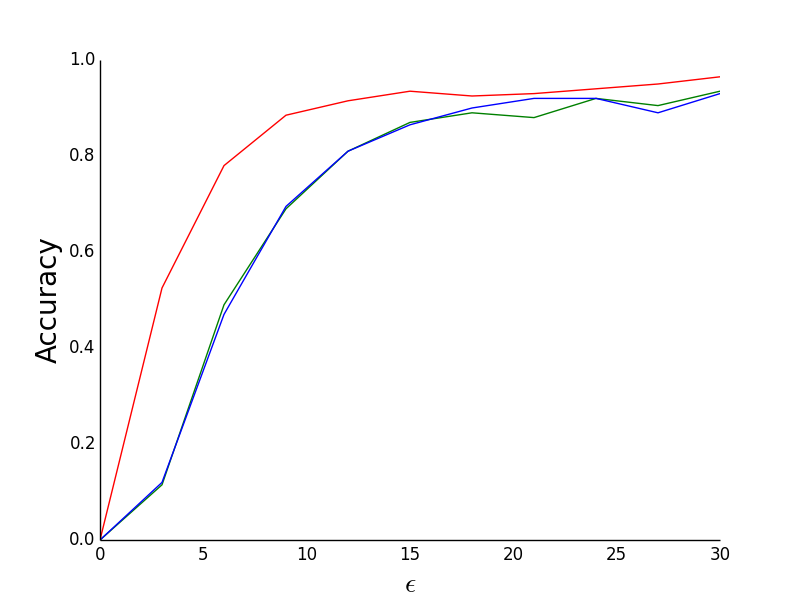}
		\caption{$m_{ret}=10$}
        \end{subfigure}
        \begin{subfigure}[B]{0.48\textwidth}
               \includegraphics[width=\textwidth]{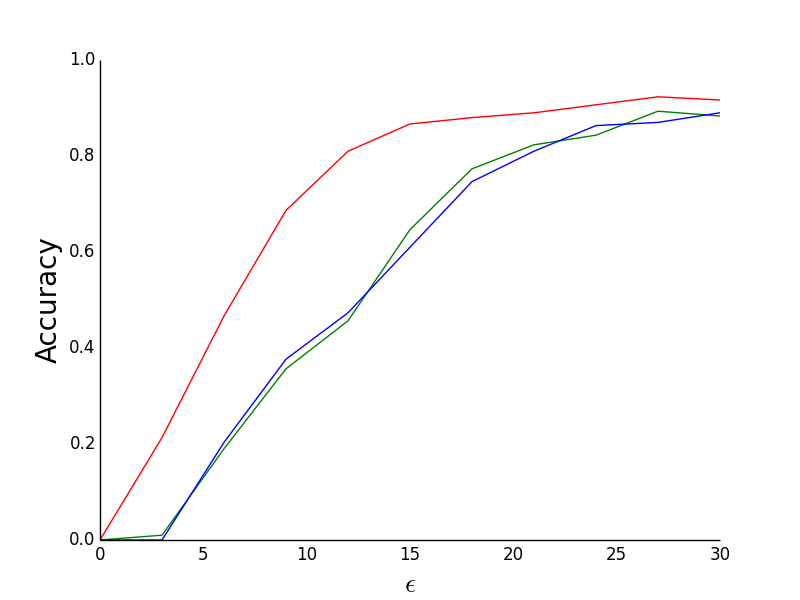}
		\caption{$m_{ret}=15$}
        \end{subfigure}
       
       \caption{We measure the accuracy (the percentage of the top SNPs correctly returned) of the three PrivSTRAT methods for picking top SNPs using score (blue), distance (red) and noise (green) based methods with $m_{ret}$ (the number of SNPs being returned) equal to (a) 10 and (b) 15  for the RA datasets, with varying values of the privacy parameter $\epsilon$. We see that, in both cases, the distance based method outperforms the others. Even still, the accuracy is fairly low for these large values of $m_{ret}$.  These results are averaged over 20 iterations.}\label{largeM}
\end{figure}

\subsection{Missing Genotype}

The above analysis assumed that there was no missing genotype data. In practice, however, many entries in a given genotype matrix will be undefined. There are various ways of dealing with this, most notably imputation. In this work we take a simpler approach (one that is built into the pysnptools package). This approach works by replacing each missing entry in the genotype vector at a given SNP with the mean value taken over all non-missing entries at that SNP. We plan for future versions of PrivSTRAT to make use of imputation based strategies for dealing with missing genotype data.

\subsection{Calculating PCA}

By default, PrivSTRAT uses an approximate version of SVD to perform the PCA in the paper, similar to that suggested in \cite{fastpca}. In particular, we use the TruncatedSVD command in sklearn.decomposition. This is due to the fact that calculating the PCA is by far the most time consuming step in the algorithm. One can, however, use an exact version of the SVD (by setting the -e flag to 1), which uses the SVD method in numpy.linalg.

\end{document}